\documentclass[onecolumn]{IEEEtran} 
\usepackage{amssymb}
\usepackage{amsmath}
\usepackage{amsfonts}
\usepackage{graphicx}
\usepackage{caption}
\usepackage{subcaption}
\usepackage{authblk}
\usepackage{amsthm}

\newif\ifpicture
\picturetrue
\newtheorem{theorem}{Theorem}

\newtheorem{example}[theorem]{Example}

\newtheorem{remark}[theorem]{Remark}

\begin{document}

\title{Nonlinear system modeling based on \\
constrained Volterra series estimates}

\author[$\dagger$]{P. \'{S}liwi\'{n}ski\footnote{This paper is a postprint of a paper
 submitted to and accepted for publication in IET Control Theory \& Applications and
 is subject to Institution of Engineering and Technology Copyright. The copy of record 
 is available at the IET Digital Library (\'{S}liwi\'{n}ski, P., et al. "Nonlinear system modeling
 based on constrained Volterra series estimates." IET Control Theory \& Applications 11.15 (2017):
2623-2629, DOI: 10.1049/iet-cta.2016.1360)}}
\author[$\star$]{A. Marconato}
\author[$\dagger$]{P. Wachel}
\author[$\star$]{G. Birpoutsoukis}

\affil[$\dagger$]{Department of Control Systems and Mechatronics, Wroc\l aw University of Science and Technology, Wybrze\.{z}e Wyspia\'{n}skiego 27, 50-370 Wroc\l aw, Poland}
\affil[$\star$]{Dept. ELEC, Vrije Universiteit Brussel, Pleinlaan 2, 1050 Brussels, Belgium}

\maketitle
\begin{abstract}
\looseness=-1 A simple nonlinear system modeling algorithm designed to work with limited \emph{a priori }knowledge and short data records, is examined. It creates an empirical Volterra series-based model of a system using an $l_{q}$-constrained least
squares algorithm with $q\geq 1$. If the system $m\left( \cdot \right) $ is a continuous and bounded map with a finite memory no longer than some known $\tau$, then (for a $D$ parameter model and for a number of measurements $N$) the difference between the resulting model of the system and the best possible theoretical one is guaranteed to be of order $\sqrt{N^{-1}\ln D}$, even for $D\geq N$. 
The performance of models obtained for $q=1,1.5$ and $2$ is tested on the
Wiener-Hammerstein benchmark system. The results suggest that the models
obtained for $q>1$ are better suited to characterize the nature of the
system, while the sparse solutions obtained for $q=1$ yield smaller error
values in terms of input-output behavior.
\end{abstract}

\section{Introduction}

We consider the following well-known problem: given limited prior knowledge
about a discrete-time nonlinear dynamic system and a limited amount of noisy
measurements $\left\{ \left( u_{n},y_{n}\right) \right\} ,$ $n=1,\ldots ,N,$
find an accurate model that describes the system's behavior; \cite{Ljun:2010}%
. The only important assumption we make on the system is that it is represented by a continuous map $m\left( \mathbf{u}_{n}\right) $, where $\mathbf{u}_{n}=[u_{n},\ldots,u_{n-\tau }]^{T}$, which has a finite memory, whose length is not specified but no larger than some $\tau \in \mathbb{N}$.

Such a system, for bounded input signals, can be approximated arbitrarily
well by a double-truncated Volterra series model; \emph{cf. e.g. }\cite%
{Alpe:1965,Boy:1984,San:1992,PeaOgu:2002,KekGia:2011}, \cite{Scho:2015} and
the recent survey \cite{Chen:2017}. In practice, when no further information
about the system is available, one would like to take the largest accessible
model (limited only by the computational resources)\ to get the best
possible approximation. Such an approach, in case of the Volterra
representation, has however an immediate consequence: the number $D$ of
model parameters is large even for models of moderate size. For the standard
least squares approach this implies that, in order to estimate these
parameters accurately from the noisy measurements of the system, the number
of measurements $N$ should be much larger than $D$. In such a case, the
corresponding computation routines become both time consuming and prone to
numerical errors; see \emph{e.g.} \cite%
{Alpe:1965,Wahl:1991,Jud95,Vand:1999,Ljun:2010}, \cite{WesKea:2003}, \cite%
{Ogu:2007}, \cite{Mar:2004}. 
The recent advances in statistics alleviate this issue by offering
constrained optimization algorithms that produce models of good quality for
a number of measurements $N$ comparable to (or even smaller than) the number
of model parameters $D$. The constraint assumptions are mild and in this
work they translate to the requirement that, for a given system, the $l_{q}$%
-norm, $q\geq 1$, of the vector composed of its Volterra representation
coefficients is finite. Our modeling algorithm relies on constrained convex
optimization techniques, \cite{Boyd:2004,Kaka:2012}, and is derived from the
aggregative algorithms which, for $q=1$, were proposed and examined in \emph{%
e.g. }\cite{JudNem:2000,Nem:2000}, and then applied to dynamic nonlinear
systems in \cite{WachSli:2015}.

The paper contribution consists in:

\begin{enumerate}
\item The class of nonlinear system modeling algorithms based on Volterra
series and constrained optimization is examined for $l_{q}$ norms, $q\geq 1$.

\item The theoretical bounds for model errors are derived for finite $D$ and 
$N$. It is shown, in particular, that these errors grow like $\sqrt{\ln D}$
and vanishes like $\sqrt{N^{-1}}$ with growing $D$ and $N$, respectively; 
\emph{cf. e.g. }\cite{Wang:2009,Zeng:2013}, where small sample size
properties are also examined in nonlinear system modeling problems.

\item The practical performance was verified using the benchmark data from 
\cite{SchSuyLju:2009} and further investigated using the numerical experiment.
\end{enumerate}

\section{Problem statement\label{Sec_PS}}

The discrete-time nonlinear system of interest is described by the following
input-output equation%
\begin{equation}
y_{n}=m\left( \mathbf{u}_{n}\right) +e_{n}=m\left( u_{n},u_{n-1},\ldots
,u_{n-\tau }\right) +e_{n},  \label{syst1}
\end{equation}%
where $u_{n}$ and $e_{n}$ are the input and the noise signals, respectively.
The system is \emph{single-input single-output} (SISO)\thinspace\ and $%
m\left( u_{n},u_{n-1},\ldots ,u_{n-\tau }\right) $ denotes a nonlinear
mapping whose output depends not only on the current input $u_{n}$ but also
on the previous $\tau $ ones.
\begin{figure}[tbp]
\centering
\par
\ifpicture
\includegraphics[width=0.25\textwidth]{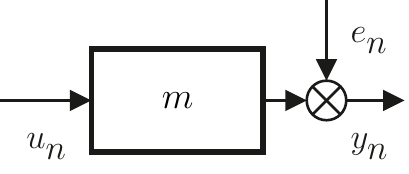} 
\fi
\caption{A nonlinear system.}
\end{figure}

About the signals and the system we assume that:

\begin{itemize}
\item[\textbf{A1}.] The system input $u_{n}$ is a sequence of bounded \emph{%
i.i.d. }random variables.

\item[\textbf{A2}.] The noise $e_{n}$ is a zero-mean \emph{i.i.d.} random
sequence with a finite variance, $\sigma ^{2}<\infty $. The noise and the
input are mutually independent sequences.

\item[\textbf{A3}] The nonlinear system $m\left( \cdot \right) $ has a finite memory of a length which is unknown but no longer than some $\tau \in 
\mathbb{N}$. Moreover, $m\left( \cdot \right) $  is 
continuous and bounded, \emph{i.e.} there is some $M_{m}>0$ for which $\left\vert m\left( \mathbf{u}_{n}\right) \right\vert \leq M_{m}$.
\end{itemize}

Assumptions \textbf{A1} and \textbf{A2} are typical for many system
identification problems; \emph{cf. e.g. }\cite{Ljun:2010}. Assumption 
\textbf{A3} says that our prior knowledge about the system is rather
limited. We neither assume the knowledge of the system structure nor imply
that there exists a specific parametric representation of any of its
elements. The continuity requirement in \textbf{A3} allows the system to be 
approximated by the Volterra series-based models; see also Remark \ref{Rem_VvsT}. 

\begin{example}
\label{Example_EFM}Any LTI system with $\tau $-finite memory satisfies
Assumption \textbf{A3}. It also holds true for the aforementioned \emph{e.g.} block-oriented cascade (Hammerstein, Wiener, Wiener-Hammerstein, \emph{etc.})\ or
multibranched cascade (\emph{e.g.} Uryson) systems with $\tau $-finite
overall memory length of their dynamic blocks and with (at least)\ Lipschitz
static nonlinearities.
\end{example}

Our goal is to find a good estimate of the system under Assumptions \textbf{%
A1}--\textbf{A3} based on a double-truncated Volterra series model

\begin{eqnarray}
V\left( \mathbf{u}_{n}\right) &=&h_{0}+H_{1L}\left( \mathbf{u}_{n}\right)
+\cdots +H_{PL}\left( \mathbf{u}_{n}\right)  \label{Volterra_model} \\
&=&h_{0}+\sum_{p=1}^{P}H_{pL}\left( \mathbf{u}_{n}\right) ,  \notag
\end{eqnarray}%
where $h_{0}$ is a constant and $H_{pL}\left( \mathbf{u}_{n}\right) $ are
the truncated Volterra operators (see \emph{e.g.} \cite%
{Alpe:1965,Boy:1984,PeaOgu:2002})%
\begin{equation*}
H_{pL}\left( \mathbf{u}_{n}\right) =\sum_{k_{1}=0}^{L-1}\cdots
\sum_{k_{p}=0}^{L-1}h_{k_{1}\ldots k_{p}}\cdot \prod_{i=1}^{p}u_{n-k_{i}},
\end{equation*}%
and where $h_{k_{1}\ldots k_{p}}$ are the $p$th order Volterra kernels; $P$
denotes the degree of the expansion and $L\leq \tau $ stands for the memory
length of the model.

The Volterra expansion is derived from the Taylor
series, see \emph{e.g.} \cite[Ch. 2.1]
{Mar04}, \cite[Ch. 1.4]{Doy02} and extends its modeling capabilities by incorporating dynamics into models. It however does not expand the modeling capability
w.r.t. nonlinearities. For the Volterra series, the class of admissible nonlinearities appears, in fact, to be smaller as indicated by an example of the \emph{peak-hold
operator} given by Boyd and Chua in \cite[p. 1152]{Boy:1984} (here in a
discrete-time version):%
\begin{equation}
y_{n}=\max_{k=0,1,\ldots }\left\{ x_{n-k}\right\} .  \label{NFM}
\end{equation}%
The operator in (\ref{NFM}) is \emph{continuous} but has no \emph{fading
memory }property and thus (by virtue of their theorem \cite[%
Th. 3]{Boy:1984})\ cannot be approximated arbitrarily well by a Volterra
series. This example shows that, in order to effectively model the nonlinear
system by the Volterra series, the assumptions imposed on the system
nonlinearity:

\begin{itemize}
\item cannot be, in general, considered separately from the dynamics, and
need to be stronger than in case of static nonlinear systems and
their Taylor expansion-based models, where, by virtue of the
Weierstrass-Stone theorem, it suffices that the nonlinearity is continuous
(if we want it to be approximated arbitrarily well - as it is the case in our paper) or is analytic (if we want to recover it fully). Nevertheless, 

\item in the particular case of systems satisfying Assumption \textbf{A3}, the continuity requirements is sufficient because of the systems' finite memory.

\end{itemize}

\begin{remark}
\label{Rem_VvsT}In a special case of block-oriented systems of known
structure (like \emph{e.g.} the Hammerstein, Wiener, Uryson, or LNL and NLN
systems), where the Volterra kernels can be expressed in terms of impulse
response coefficients of dynamic linear blocks and of derivatives of the
static nonlinearities, the constraints imposed on these nonlinearities can
also be equivalent to the Taylor series ones; \emph{cf. e.g.} \cite[Ch. 2.1.1]%
{Mar04} and \cite{KekGia:2011}. 
\end{remark}


\section{The algorithm\label{Sec_Alg}}

Thanks to the fact that the kernels $h_{k_{1}\ldots k_{p}}$ in (\ref%
{Volterra_model}) are symmetric with respect to the permutations of indices $%
k_{1}\ldots k_{p}$, the number of parameters $D$ to be estimated reduces
from $\frac{L^{P+1}-1}{L-1}$ to $\binom{L+P}{L}$; \emph{cf. }\cite%
{KekGia:2011}.\ However, even for small $P$ and $L$, $D$ is a large number
and might not be further reduced without additional information about the
structure of the system; see \emph{e.g. }\cite{KibFav:2006,KekGia:2011}.
Moreover, even if the system structure was known, we would not be able to
decide what parameters $P$ and $L$ should be used since, under Assumption 
\textbf{A3}, the class of systems is still too general to be exactly
represented by Volterra models. Therefore, in such a scenario one would take
the largest available dictionary and let the data select the best subset.
This requires an algorithm which either:

\begin{itemize}
\item selects the largest (\emph{i.e.} the most important)\ coefficients\ (%
\cite{Tibs:1996,Kukr:2009,JameTibs:2013}) or, at least,

\item is robust against the effect of overparametrization (\emph{i.e.} the
excessive number of model parameters).
\end{itemize}

We will show that the presented algorithm offers the latter property and, in
particular, works when $D>N$.

\begin{remark}
Some system structures, \emph{e.g.} the Hammerstein or Uryson ones,\ have
structured or sparse Volterra representations \cite{KibFav:2006,KekGia:2011}%
. Nevertheless, we do not assume that the structure of the system is known
and therefore we need an algorithm which will work in either case (note that
the LASSO-based\ algorithms are designed to work when it is assumed \emph{a
priori }that the model has a sparse structure \cite{JameTibs:2013}).
\end{remark}

For notation simplicity, we arrange the coefficients of the Volterra series
kernels $h_{k_{1}\ldots k_{p}},$ $p=1,\ldots ,P$, $k_{j}=0,\ldots ,L-1$ and $%
j=0,\ldots ,p$, from the model (\ref{Volterra_model}), into a vector $%
\mathbf{\theta }=\left[ \theta _{1},\theta _{2},\dots ,\theta _{D}\right]
^{T}$ and denote the corresponding Volterra terms as $m_{i}\left( \mathbf{u}%
\right) ,$ $i=1,\ldots ,D$. The collection $\left\{ m_{i}\left( \mathbf{u}%
\right) \right\} $ will be referred to as a dictionary. Note that by
Assumption \textbf{A1} we have that $\left\vert m_{i}\left( \mathbf{u}%
\right) \right\vert \leq M_{d}$, for some $M_{d}>0.$

The system model is thus expressed as%
\begin{equation}
\hat{m}(\mathbf{u};\mathbf{\hat{\theta})}=\sum_{i=1}^{D}\hat{\theta}%
_{i}m_{i}\left( \mathbf{u}\right) ,  \label{model}
\end{equation}%
where $\mathbf{\hat{\theta}}=[\hat{\theta}_{1},\hat{\theta}_{2},\dots ,\hat{%
\theta}_{D}]^{T}$ is the empirical counterpart of the vector $\mathbf{\theta 
}$, obtained from the measurement set $\left\{ \left( u_{n},y_{n}\right)
\right\} ,$ $n=1,\ldots ,N$, by minimization of the empirical quadratic
criterion (note that the summation in (\ref{Q_rand}) starts at $\tau +1$
since the first $\tau $ values of the input are not known):%
\begin{equation}
\hat{Q}\left( \mathbf{\theta }\right) =\frac{1}{N-\tau }\sum_{i=\tau +1}^{N}%
\left[ \hat{m}\left( \mathbf{u}_{i};\mathbf{\theta }\right) -y_{i}\right]
^{2}.  \label{Q_rand}
\end{equation}%
with the following constraint imposed on the solution%
\begin{equation}
\left\Vert \mathbf{\theta }\right\Vert _{q}\leq D^{\frac{1}{q}-1}\text{, for
a given }q\geq 1.  \label{argmin}
\end{equation}

\begin{remark}
\label{remark_sparsity}By virtue of the well-known relation between the $%
l_{q}$-norms%
\begin{equation}
\left\Vert \mathbf{\theta }\right\Vert _{q}\leq \left\Vert \mathbf{\theta }%
\right\Vert _{1}\leq D^{1-\frac{1}{q}}\left\Vert \mathbf{\theta }\right\Vert
_{q},  \label{norm_ineq}
\end{equation}%
the constraint in (\ref{argmin}) implies that $\left\Vert \mathbf{\theta }%
\right\Vert _{1}\leq 1$, which is a pivotal fact used in the proof of the
theoretical behavior of the model $\hat{m}(\mathbf{u};\mathbf{\hat{\theta})}$%
, and cannot be relaxed in general. However, if it is known that the
structure of the system has a sparse Volterra representation, that is, if $%
\mathbf{\theta }$ is $K$-sparse\footnote{%
That is, there are $K$ non-zero coefficients.} for some $K<D$, then the
right-hand side inequality in (\ref{norm_ineq}) turns into $\left\Vert 
\mathbf{\theta }\right\Vert _{1}\leq K^{1-\frac{1}{q}}\left\Vert \mathbf{%
\theta }\right\Vert _{q}$. In this case, the constraint in (\ref{argmin})
remains the same for $q=1$, but for $q>1$ it can be turned into a weaker one%
\begin{equation*}
\left\Vert \mathbf{\theta }\right\Vert _{q}\leq K^{\frac{1}{q}-1}.
\end{equation*}
\end{remark}

\begin{example}
For $q=1,1.5$ and $2$, the constraint (\ref{argmin}) takes the forms 
\begin{equation*}
\left\Vert \mathbf{\theta }\right\Vert _{1}\leq 1\text{, }\left\Vert \mathbf{%
\theta }\right\Vert _{1.5}\leq \sqrt[3]{D^{-1}}\text{, and }\left\Vert 
\mathbf{\theta }\right\Vert _{2}\leq \sqrt{D^{-1}},
\end{equation*}%
respectively.
\end{example}

\begin{remark}
In case of Volterra series, the constraint in (\ref{argmin}) means that the
sequence of Volterra kernel coefficients $\theta $ is in the space $l_{q}$,
that is, they are $l_{q}$-summable. Note that this assumption is satisfied
for all models with finite $P$\ and $L$\ (and for all $q\geq 1$) up to the
multiplicative constant -- see Section \ref{section_tuning} and \emph{cf.} Assumption \textbf{A3}. In turn, for
the infinite memory models, it is satisfied if the system possesses a fading
memory property; see \cite{Boy85}. Such a property holds, \emph{e.g. }for
Hammerstein, Wiener, Hammerstein-Wiener systems with Lipschitz
nonlinearities and asymptotically stable linear subsystems; \emph{cf. }%
Remark \ref{Rem_VvsT} and Example \ref{Example_EFM}.
\end{remark}

\begin{figure}[tbp]
\centering
\par
\ifpicture
\includegraphics[width=0.5\textwidth]{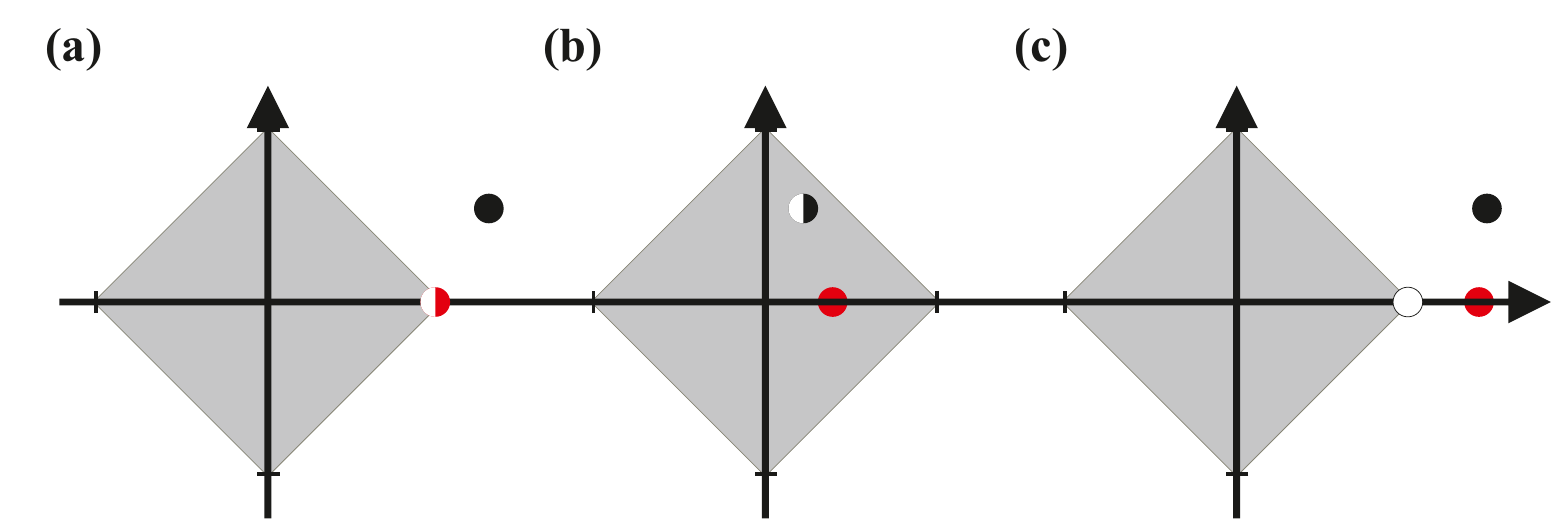} \fi
\caption{Possible selections of the tuning parameter (multiplication factor $%
R$) for $q=1$: \textbf{(a) }the optimal one (the smallest for which $%
\left\Vert \mathbf{\protect\theta }^{\ast }\right\Vert _{1}=1$ holds); 
\textbf{(b)} too large $R$ can make the unconstrained LS solution
acceptable; \textbf{(c)} too small $R$ can make the optimal solution $%
\mathbf{\protect\theta }^{\ast }$ infeasible (black dot -- LS solution, red
-- the optimal one $\mathbf{\protect\theta }^{\ast }$, and white -- the $%
l_{1}$ constrained empirical solution $\mathbf{\hat{\protect\theta}}$).}
\label{Fig_MF}
\end{figure}

\subsection{Theoretical properties}

Here we present the main result of the paper, that is, the upper bound of
the discrepancy between the empirical models (\ref{model}) and the best
possible one, given as%
\begin{equation}
\hat{m}\left( \mathbf{u};\mathbf{\theta }^{\ast }\right)
=\sum_{i=1}^{D}\theta _{i}^{\ast }m_{i}\left( \mathbf{u}\right) ,
\label{conv_opt}
\end{equation}%
where $\mathbf{\theta }^{\ast }$ is the solution to the following
constrained least squares problem%
\begin{equation}
\mathbf{\theta }^{\ast }=argmin_{\lVert \mathbf{\theta }%
\rVert _{q}\leq D^{1/q-1}}Q\left( \mathbf{\theta }\right) ,
Q\left( 
\mathbf{\theta }\right) =E\left\{ \hat{m}\left( \mathbf{u}_{n};\mathbf{%
\theta }\right) -y_{n}\right\} ^{2},  \label{Q_opt}
\end{equation}%
for a given $q\geq 1$; \emph{cf.} (\ref{Q_rand}) and (\ref{argmin}).

The theorem below gives the upper bound of this discrepancy and offers a
formal justification of a good behavior of the examined models in case when $%
D$ is large and $D\geq N$.

\begin{theorem}
\label{Th-1}Let the nonlinear system (\ref{syst1}) fulfill Assumptions 
\textbf{A1}--\textbf{A3} and let (\ref{argmin}) hold for some $q=\chi \geq 1$%
, then, for any $q\in \left[ 1,\chi \right] $, the difference between the
empirical model, $\hat{m}(\mathbf{u};\mathbf{\hat{\theta})}$, and the best
possible one, $m\left( \mathbf{u};\mathbf{\theta }^{\ast }\right) $, has the
following upper bound%
\begin{equation}
E\{Q(\mathbf{\hat{\theta})\}}-Q\left( \mathbf{\theta }^{\mathbf{\ast }%
}\right) \leq C\cdot \frac{\sqrt{N}}{N-\tau }\sqrt{\left( \tau +1\right) \ln
D},  \label{CV_psi_II}
\end{equation}%
\newline
for any $D>2,$ where $C=32\sqrt{e}(M\sigma +2M^{2})$ and $M=\max \left\{
M_{m},M_{d}\right\} $.
\end{theorem}

\begin{proof}
See Appendix.
\end{proof}

The theorem generalizes the result obtained in \cite{WachSli:2015} for $q=1$%
\ (see also the original one in \cite{JudNem:2000}) and, in particular, says
that:

\begin{itemize}
\item The upper bound of the error is practically immune to the number of
model parameters $D$ as it grows only logarithmically with $D$. This
property is of special significance for the Volterra series-based models,
for which $D$ grows fast with both $L$ and $P$.

\item For $D$ being of order equal to (or larger than) $\sqrt{N}$ (and hence
for $D\geq N$), our models have the error upper bound lower than those
produced by the unconstrained least squares algorithms, for which this bound
is of order $D/N$; \emph{cf. e.g.} \cite{Ljun:2010}.
\end{itemize}

\begin{remark}
\label{Rem_Gamma}The exact value of $\tau $\ is only needed to establish the
formal bound of the error in (9). If $\tau $ is not known (or known to be
infinite), then one should expect slightly worse performance of the
algorithm -- as shown for $q=1$ in \cite{Wach:2016} that if $\tau =\infty $,
then the error bound vanishes slower (by a factor $\sqrt{\ln N}$) with
growing $N$ and the additional term $N^{-c},c>0$\ occurs. 
\end{remark}

\begin{remark}
Assumptions \textbf{A1-A3} can be slightly weakened. For instance, one can
admit correlated input; see \emph{e.g.} \cite{WachSli:2015} for $q=1$. In
that case, the resulting models error vanishes with growing $N$ and still
maintain its robustness against large number of model parameters $D$ as in (%
\ref{CV_psi_II}). The only drawback is that the overall memory length $\tau $
need to be increased in order to take into account the correlation of the
signal. Also, the system memory does not need to be finite, however, the
resulting error bound is larger as indicated in the Remark \ref{Rem_Gamma}.
\end{remark}

\subsection{Tuning the algorithm\label{section_tuning}}

Usually, we do not know \emph{a priori }whether the constraint (\ref{argmin}%
) is satisfied by the system of interest. Observe however that we can make
any system satisfying the Assumption \textbf{A3} compliant with (\ref{argmin}%
) by multiplying the dictionary entries by some, sufficiently large, factor $%
R>0$.\footnote{%
Or by using the equivalent constraint $\Vert \hat{\theta}\Vert _{q}\leq
R\cdot D^{1/q-1}.$} Below we shortly examine the impact of $R$ on the
aggregation error bound in (\ref{CV_psi_II}). 
\begin{figure}[tbp]
\centering
\par
\ifpicture
\includegraphics[width=0.5\textwidth]{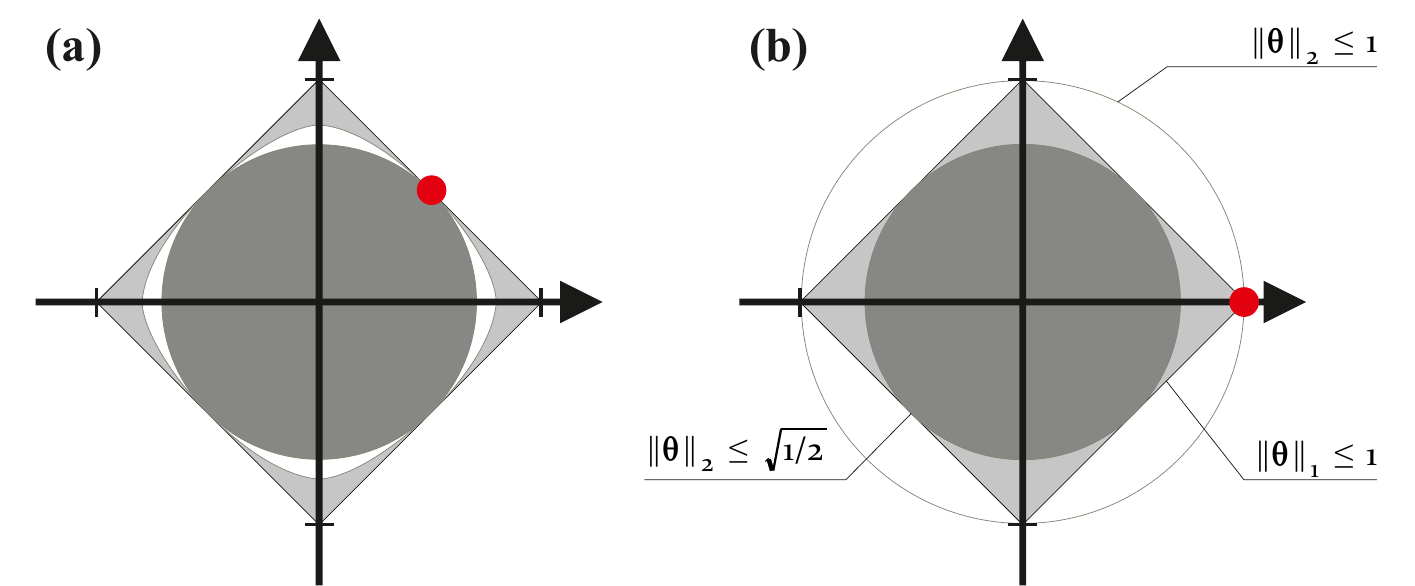} \fi
\caption{\textbf{(a) }Shapes and relative sizes of the constraint in (%
\protect\ref{argmin}) for $q=1$ (light grey), $q=1.5$ (white) and $q=2$
(dark gray); $D=2$; $\mathbf{\protect\theta }^{\ast }=[1/2~1/2]$ \textbf{(b)}
For a system with a 'sparse' representation ($\mathbf{\protect\theta }^{\ast
}=[1~0]$), the dictionary entries, which are optimal for $q=1$, need to be,
for $q=2$, multiplied by $R=\protect\sqrt{2}$ (which is also equivalent to
multiplying by $R$ the bound in (\protect\ref{argmin})).}
\label{Fig_shapes}
\end{figure}

\subsubsection{Case $q=1$}

For the $l_{1}$ constraint, $R$ has a quite straightforward impact on the
aggregation error bound (\ref{CV_psi_II}): its squared value occurs there as
a multiplicative factor (see (\ref{MR})\ in Remark \ref{MR_proof} following
the proof in Appendix \ref{Th-1} and \emph{cf.} (\ref{CV_psi_II})):%
\begin{equation}
C=32\sqrt{e}(RM\sigma +2\left( RM\right) ^{2})  \label{CV_psi_III}
\end{equation}%
One would therefore like to make $R$ small, however, decreasing $R$ pushes
the optimal solution $\mathbf{\theta }^{\mathbf{\ast }}$ towards the
constraint boundary and can eventually set it outside of it; see the
illustrations in Figs. \ref{Fig_MF}a-c.

\subsubsection{Case $1<q\leq 2$}

The left-hand side of the norm inequality in (\ref{norm_ineq}) indicates
(and Fig. \ref{Fig_shapes}a illustrates it for $D=2$) that if for a given $%
q=\chi >1$, the constraint in (\ref{argmin}) holds true for a system, then
it holds for any other $q$'s such that $1\leq q<\chi $ and -- consequently
-- the upper bound constants in (\ref{CV_psi_III}) remain valid for all
these $q$'s.

In the reversed situation, however, when it is only known that (\ref{argmin}%
) holds for $q=1$, one needs to take $R=D^{1-\frac{1}{q}}$\ (\emph{cf.} the
right-hand side of the inequality in (\ref{norm_ineq}) and Fig. \ref%
{Fig_shapes}b) to assure that (\ref{argmin}) is valid for $q>1$ as well.
This, unfortunately, can have a detrimental influence on the algorithm
behavior since then%
\begin{equation}
C=32\sqrt{e}(D^{1-\frac{1}{q}}M\sigma +2(D^{1-\frac{1}{q}}M)^{2})
\label{CV_psi_IV}
\end{equation}%
and the upper bound of the aggregation error increases with $q$ and becomes
approximately $D^{2(1-1/q)}$ times larger (\emph{i.e. }up to $D$ times
larger for $q=2$) than in (\ref{CV_psi_II})\footnote{%
In case of a $K$-sparse representation, the factor $D^{1-1/q}$ in (\ref%
{CV_psi_IV}) reduces to $K^{1-1/q}$.}.

\subsubsection{Empirical tuning algorithm\label{Algorithm_tuning}}

Below we propose a simple empirical algorithm to select a value of the
tuning parameter $R$.

{\textbf{Algorithm:}} For a given $q\geq 1$, pick some (large) $R>0$, such
that $\Vert \hat{\theta}\Vert _{q}\ll D^{\frac{1}{q}-1}$. Next, pick some
(small) $\varepsilon >0$ and decrease\footnote{%
It can effectively be implemented using a fast bisection search procedure.} $%
R$ until 
\begin{equation*}
\Vert \hat{\theta}\Vert _{q}\in \lbrack D^{\frac{1}{q}-1}-\varepsilon ,D^{%
\frac{1}{q}-1}).
\end{equation*}

In the presence of noise one should usually decrease $R$ even more once $%
\Vert \hat{\theta}\Vert _{q}=D^{\frac{1}{q}-1}$ is attained, however, there
is a risk that in such a case the actual solution $\theta ^{\ast }$ will be
shifted outside the constraint and a systematic (bias) error will be
introduced; see Fig. \ref{Fig_MF}c.

\section{Experimental/simulation results\label{section_experiments}}

We first tested the algorithm on a real data example, namely the
Wiener-Hammerstein benchmark presented in \cite{SchSuyLju:2009}.\footnote{%
Since the data are generated by the real system, we cannot assure that all
the assumptions \textbf{A1-A3} hold.} Then, to get a deeper insight into the
behavior of the algorithm in a controlled environment, we also made
additional experiments using a 'toy-example' system with a
Wiener-Hammerstein structure; the systems are referred to as WHB\ and WH2,
respectively.

\subsection{Wiener-Hammerstein benchmark}

The benchmark data were taken from a nonlinear electronic system with a
Wiener-Hammerstein structure, as depicted in Fig.~\ref{figWH}. More details
on how the benchmark data were generated can be found in \cite%
{SchSuyLju:2009}.

\begin{figure}[t]
\centering
\ifpicture
\includegraphics[width=0.5\textwidth]{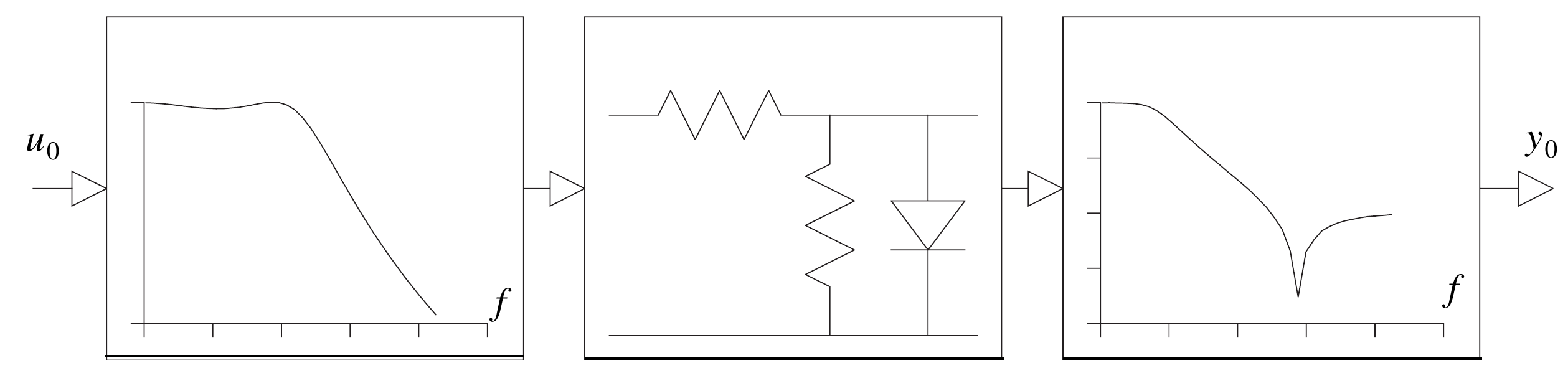} \fi
\caption{The Wiener-Hammerstein benchmark system.}
\label{figWH}
\end{figure}

The system was modeled with a third degree Volterra series with varying
(decreasing) memory lengths of each kernel, namely $L_{1}=80,$ $L_{2}=40$
and $L_{3}=20$. It resulted in the overall number of parameters $D=2441$.
The measurement sets of lengths $N=500,1000,2000$ and $5000$ were taken from
the benchmark's learning data pool. Three models were estimated for $q=1,1.5$
and $2$. Algorithm \ref{Algorithm_tuning} was used to tune the $R$ factor
separately for each $q$; see Fig. \ref{Fig_shapes}.

\begin{remark}
Decreasing the memory length in the consecutive Volterra kernels was
dictated by the excessive computational overhead (the model with equal
memory length kernels would have $\binom{80+3}{3}=\allowbreak 91\,881$
parameters) and the numerical issues encountered for the model of that size. 
\begin{figure}[tbp]
\centering%
\begin{subfigure}{0.5\textwidth}
\ifpicture
        \includegraphics[width=0.8\textwidth]{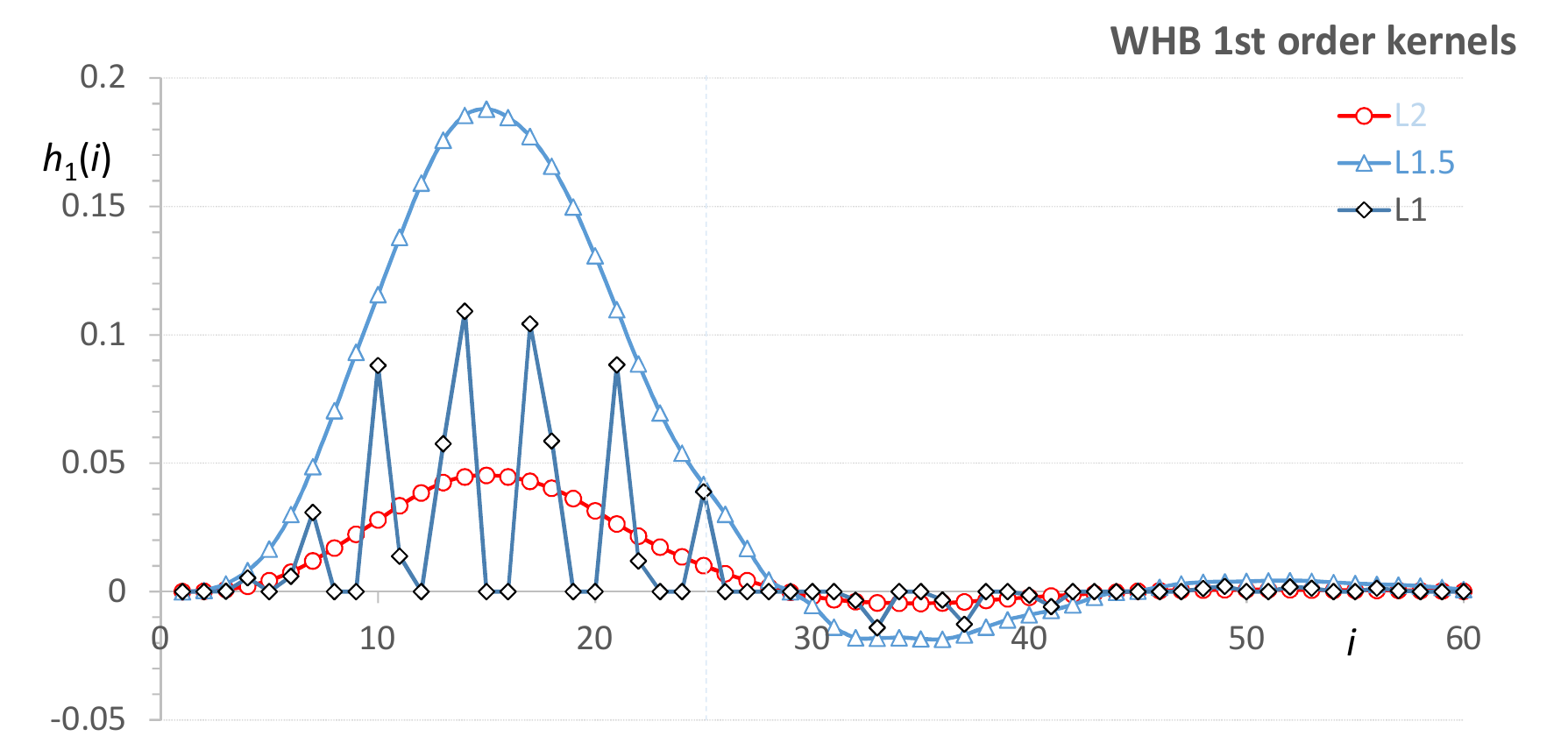}
\fi
        \caption{First order Volterra kernels for $q=1,1.5,2$.}
    \end{subfigure}
\par
\begin{subfigure}{0.5\textwidth}
\ifpicture
       \includegraphics[width=0.8\textwidth]{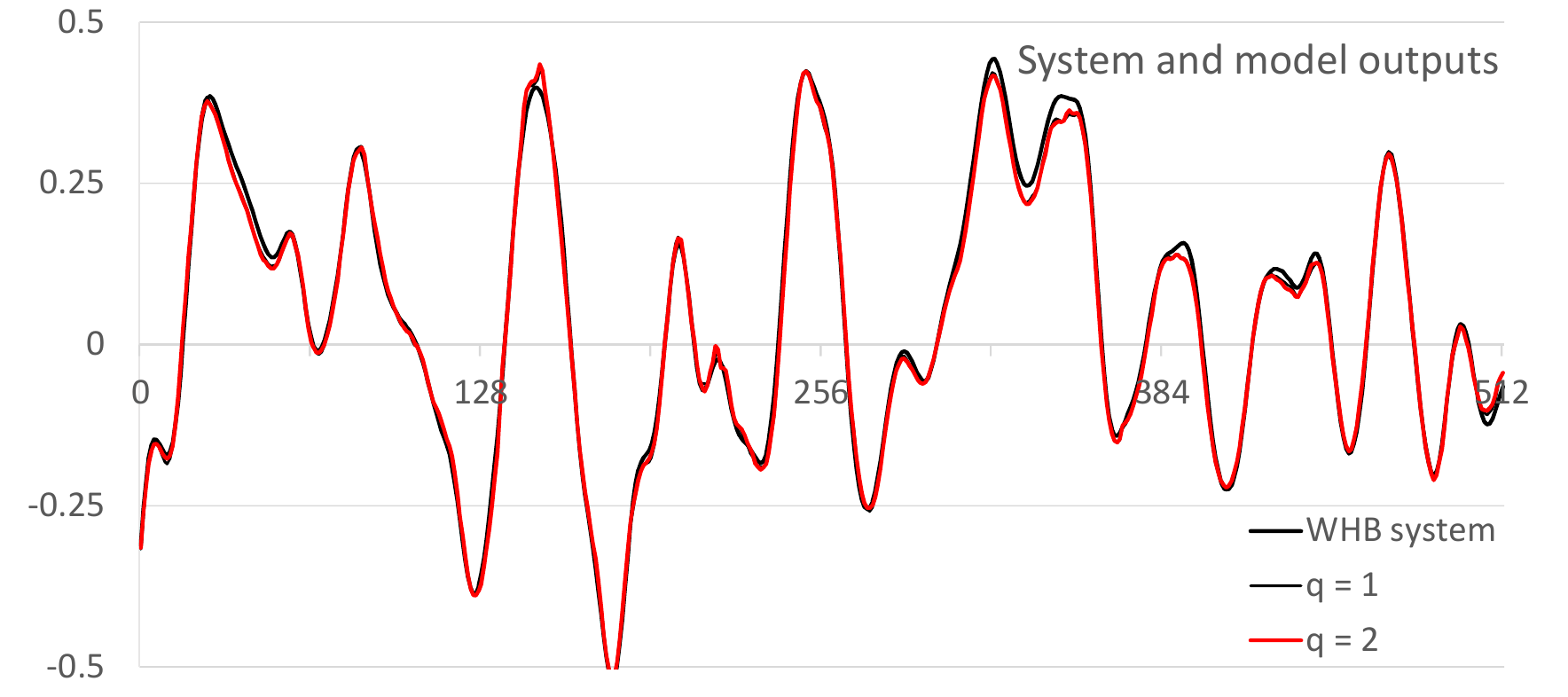}
\fi
        \caption{Outputs of the system and its models for $q=1,2$.}
    \end{subfigure}
\caption{WHB benchmark. Model kernels shapes and outputs comparison.}
\label{FIG_WHB_shapes}
\end{figure}
\end{remark}

The experiments revealed that for each $N$ the model obtained for $q=1$
offered the smallest RMSe error (calculated according to the formula (4) in 
\cite{SchSuyLju:2009}) with respect to input-output behavior; see Fig. \ref%
{Fig_WHB_errors}. It can be explained by its tendency to select a sparse
representation; see Fig. \ref{FIG_WHB_shapes} and \ref{Fig_WH2_shapes}; 
\emph{cf. e.g.} \cite{Will:2012,Marc:2012} and \cite[Ch. 6]{JameTibs:2013}. 
\begin{figure}[tbp]
\centering%
\begin{subfigure}{0.5\textwidth}
\ifpicture
	\includegraphics[width=0.8\textwidth]{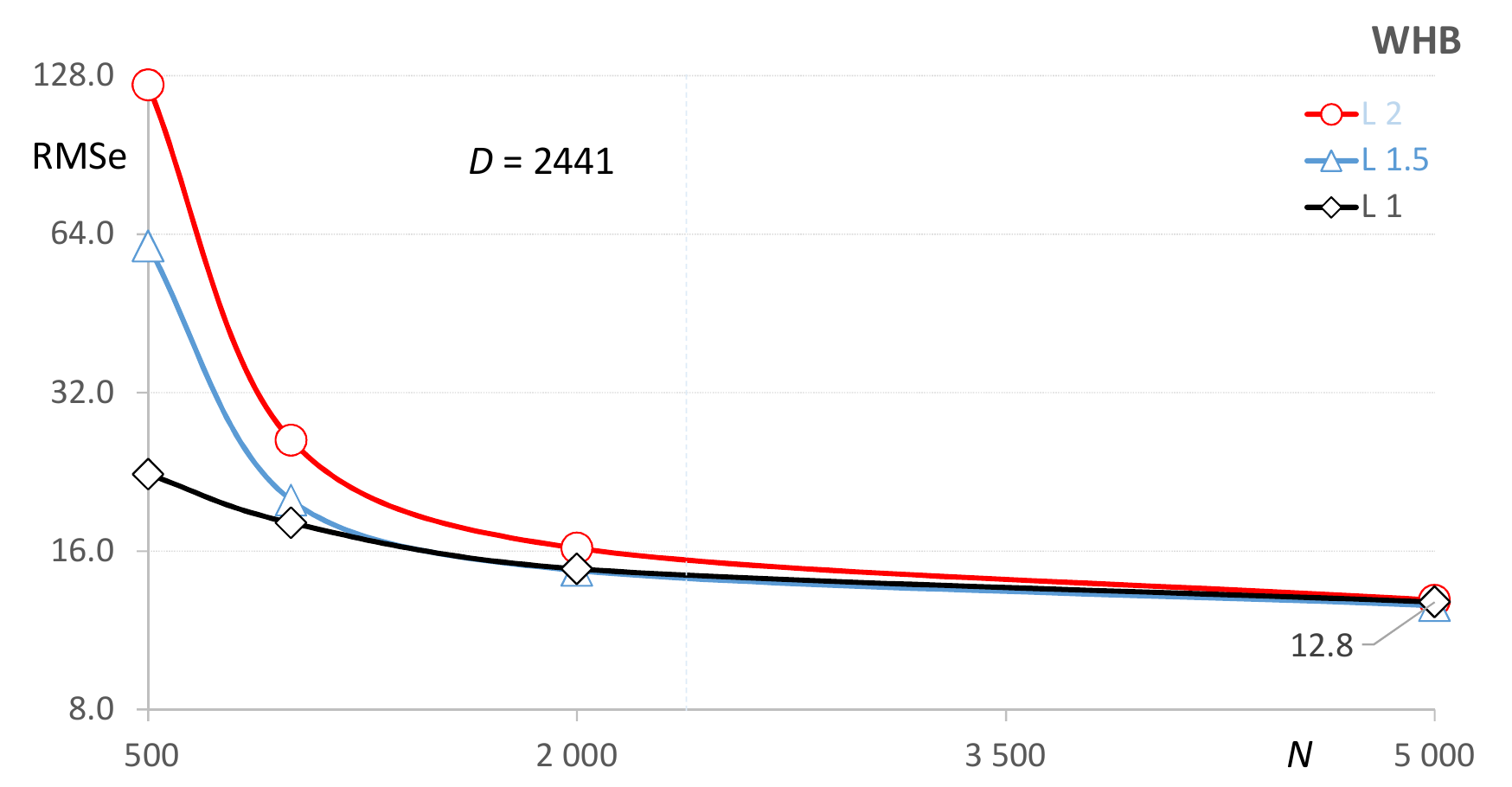}
\fi
        \caption{Model errors}
    \end{subfigure}
\par
\begin{subfigure}{0.5\textwidth}
\ifpicture
        \includegraphics[width=0.8\textwidth]{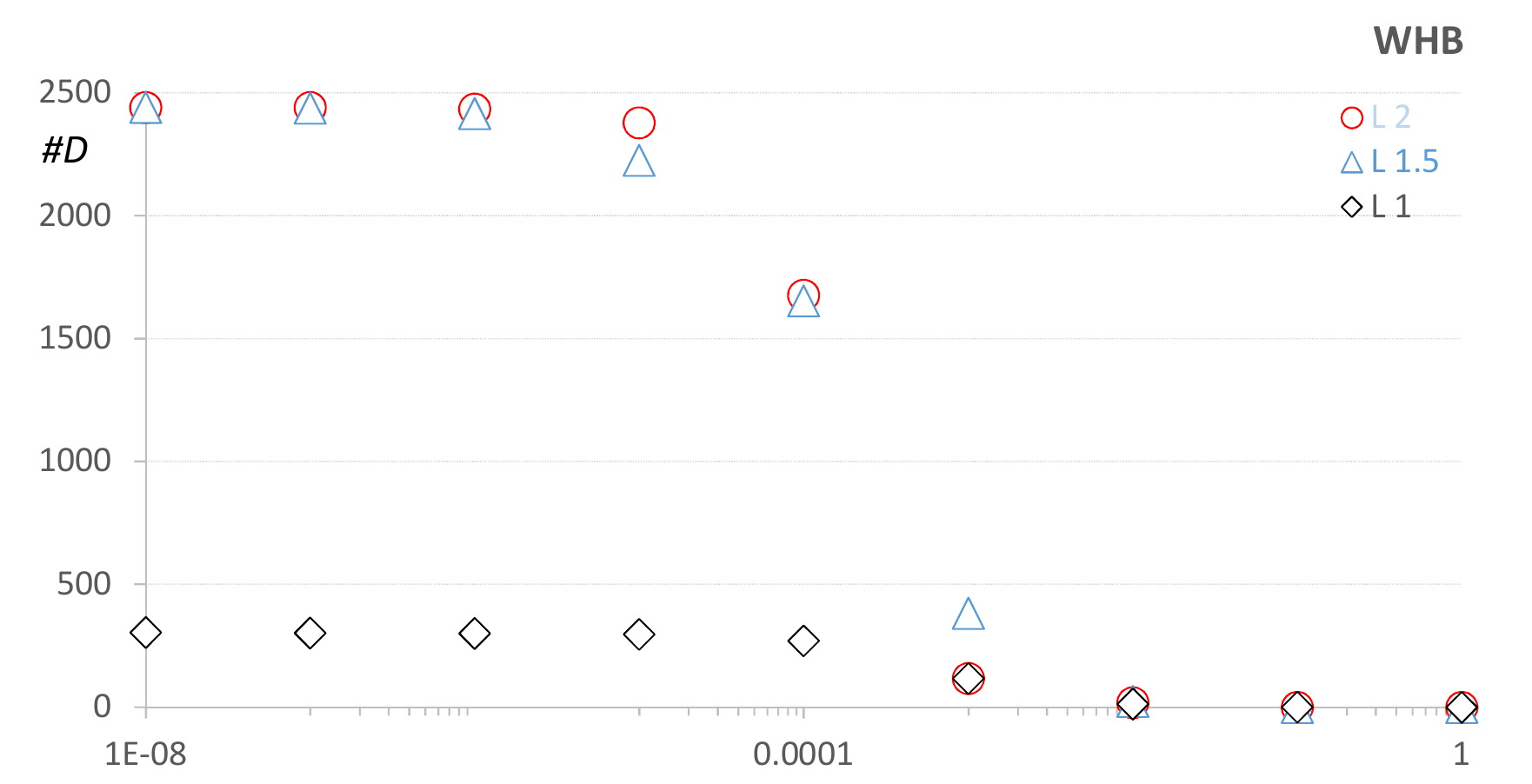}
\fi
        \caption{Model sparsities, i.e. the number of model coefficients greater than value on the $OX$ axis.}
    \end{subfigure}
\caption{WHB benchmark. Comparison of model errors and their sparsities}
\label{Fig_WHB_errors}
\end{figure}

\subsection{Simulation example}

Let us now consider a system with the same Wiener-Hammerstein structure, but
with the nonlinearity $x^{2}$, the input dynamics having the two-tap impulse
response $\lambda _{0}=1$ and $\lambda _{1}=-1,$ and with the output
dynamics having the following discrete transfer function $G\left( z\right)
=0.2655z/\left( z^{2}-1.714z+0.78\right) $. The input signal $\left\{
u_{n}\right\} $ was white and uniformly distributed in the interval $[-\sqrt{%
3},\sqrt{3}]$. The additive output noise was white and Gaussian, and was
scaled in order to make the SNR = $1,10,40,80$ and $100$. The measurement
sets had lengths $N=500,750,1000,1250,1500$ and $2000$. The Volterra model
used in this experiment was the second order one with equal memory length
kernels, $L_{1}=L_{2}=L=40$ (it thus had $D=861$ parameters).

The experiment was repeated to evaluate three models, for $q=1,1.5$ and $2$,
respectively. The results confirm the advantage of the model obtained for $%
q=1$ which, regardless of the noise level, is still able to yield the
sparsest model of the system and the lowest RMSe values; see Fig. \ref%
{Fig_WH2_errors}a,b and \emph{cf. }Fig. \ref{Fig_WHB_errors}a,b.

\begin{figure}[tbp]
\centering%
\begin{subfigure}{0.5\textwidth}
\ifpicture
        \includegraphics[width=0.8\textwidth]{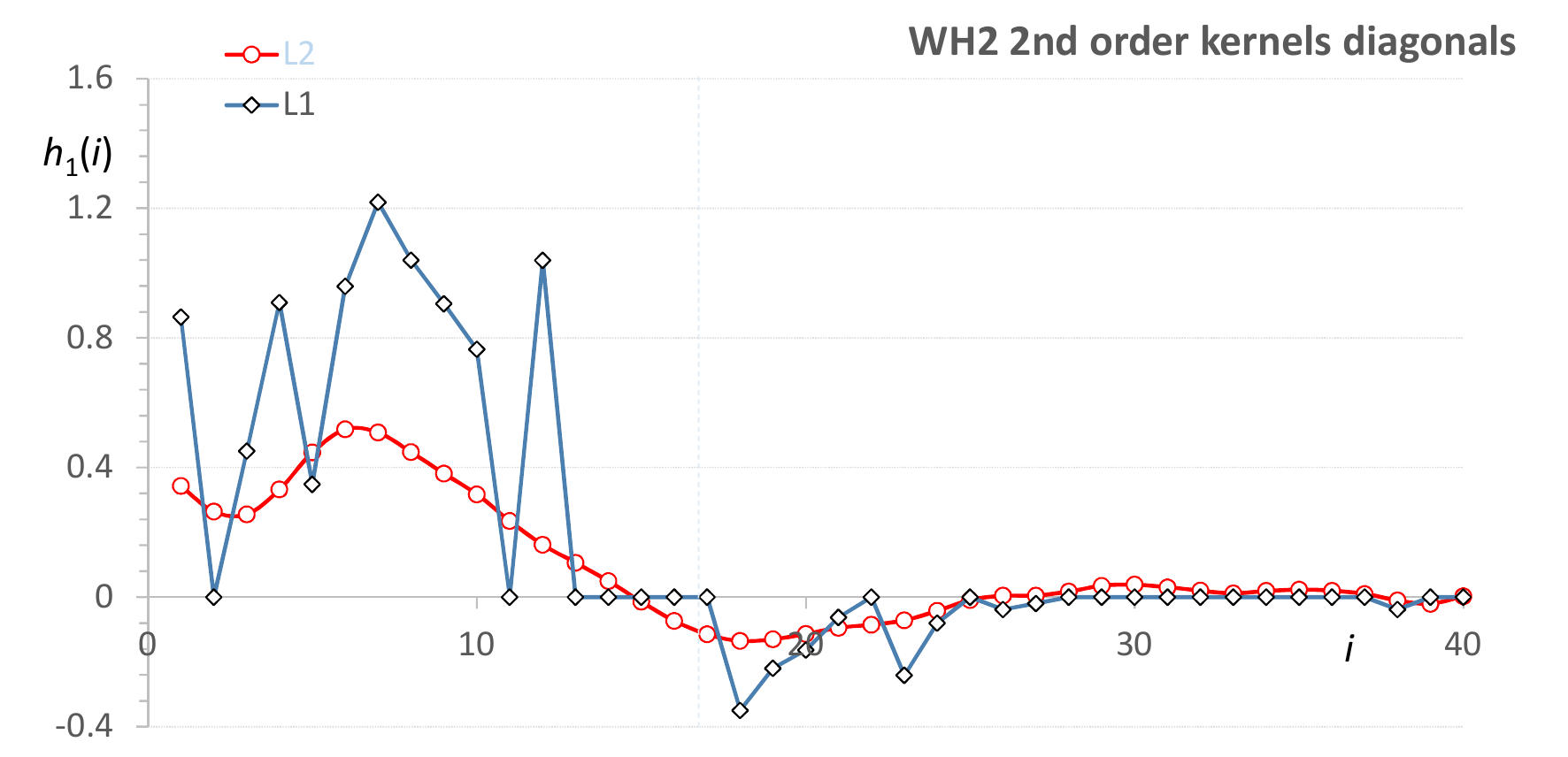}
\fi
        \caption{Second order Volterra kernels diagonals in models for $q=1,1.5,2$.}
    \end{subfigure}
\par
\begin{subfigure}{0.5\textwidth}
\ifpicture
        \includegraphics[width=0.8\textwidth]{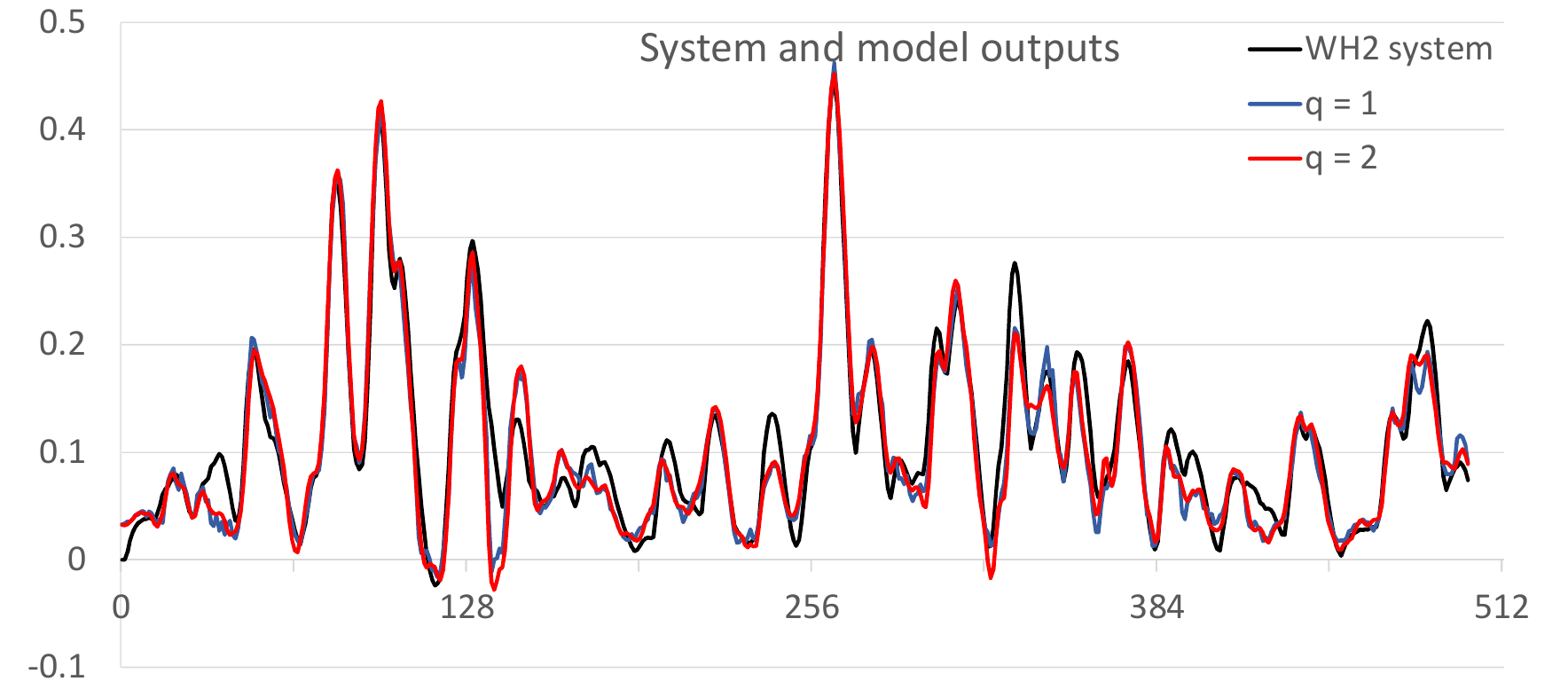}
\fi
        \caption{Outputs of the system and its models for $q=1,2$.}
    \end{subfigure}
\caption{WH2 system driven by a correlated input signal.}
\label{Fig_WH2_shapes}
\end{figure}

Comparing the diagrams in Figs. \ref{FIG_WHB_shapes}a,b one can observe that
-- in spite of various shapes of the kernels obtained for various $q$ -- the
WHB system models produced almost the same output signals. In attempt to
replicate this phenomenon for the WH2 system we correlated its input signal%
\footnote{%
In \cite{WachSli:2015} the correlated input was admitted and the result in
that work can easily be applied to our case.} by the dynamics $G\left(
z\right) $ and removed the output noise; see Figs. \ref{Fig_WH2_shapes}a,b.
Now, to explain this behavior, we denote by $\mathbf{\hat{\theta}}_{1}$, $%
\mathbf{\hat{\theta}}_{2}$ the solutions found by the algorithm for $q=1$
and for $q=2$, and by $\mathbf{S}$ the matrix composed from the consecutive
Volterra terms%
\begin{eqnarray}
\mathbf{S} &=&[\mathbf{s}_{1}^{T},\mathbf{s}_{2}^{T},\allowbreak \dots
,\allowbreak \mathbf{s}_{N}^{T}]^{T},  \notag \\
\mathbf{s}_{n} &=&[m_{1}\left( \mathbf{u}_{n}\right) ,\allowbreak
m_{2}\left( \mathbf{u}_{n}\right) ,\allowbreak \dots ,\allowbreak
m_{D}\left( \mathbf{u}_{n}\right) ]^{T},  \label{sh}
\end{eqnarray}%
The different shapes of the kernels, together with virtually the same
outputs produced by either model, mean therefore that $\mathbf{S}\left( 
\mathbf{\hat{\theta}}_{1}-\mathbf{\hat{\theta}}_{2}\right) \approx \mathbf{0}
$, which implies that the input signal correlation made the columns of $%
\mathbf{S}$ linearly dependent (\emph{i.e.} the difference vector $\mathbf{%
\hat{\theta}}_{1}-\mathbf{\hat{\theta}}_{2}$ belongs to the null space of
that matrix);\emph{\ cf.} also \cite{Hebi:2013}, where the LASSO\ algorithm
was tested for correlated data.

\begin{figure}[tbp]
\centering%
\begin{subfigure}{0.5\textwidth}
\ifpicture
        \includegraphics[width=0.8\textwidth]{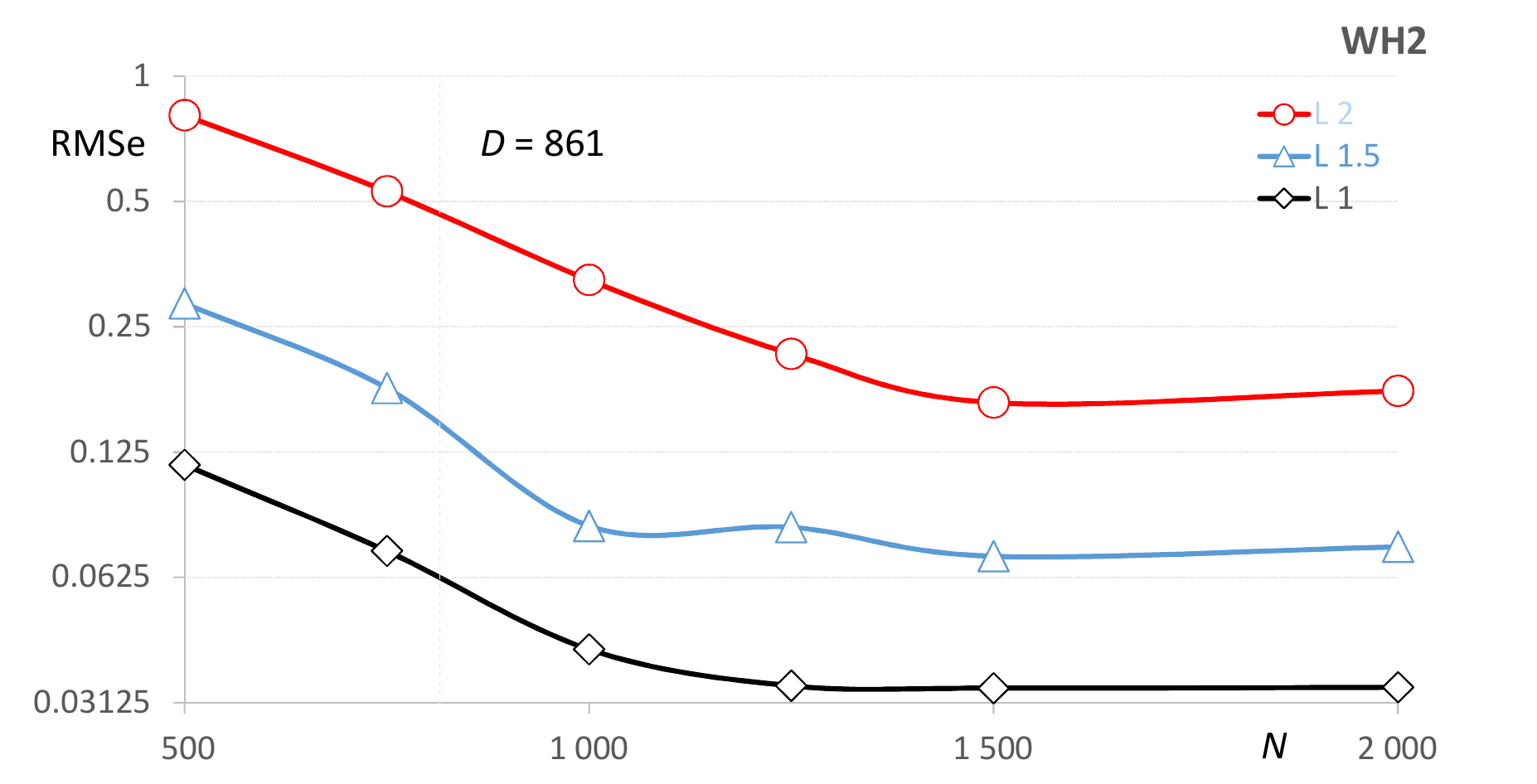}
\fi
\caption{Model errors (averaged over SNR = 1, 10, 40, 80 and 100).}
    \end{subfigure}
\par
\begin{subfigure}{0.5\textwidth}
\ifpicture
        \includegraphics[width=0.8\textwidth]{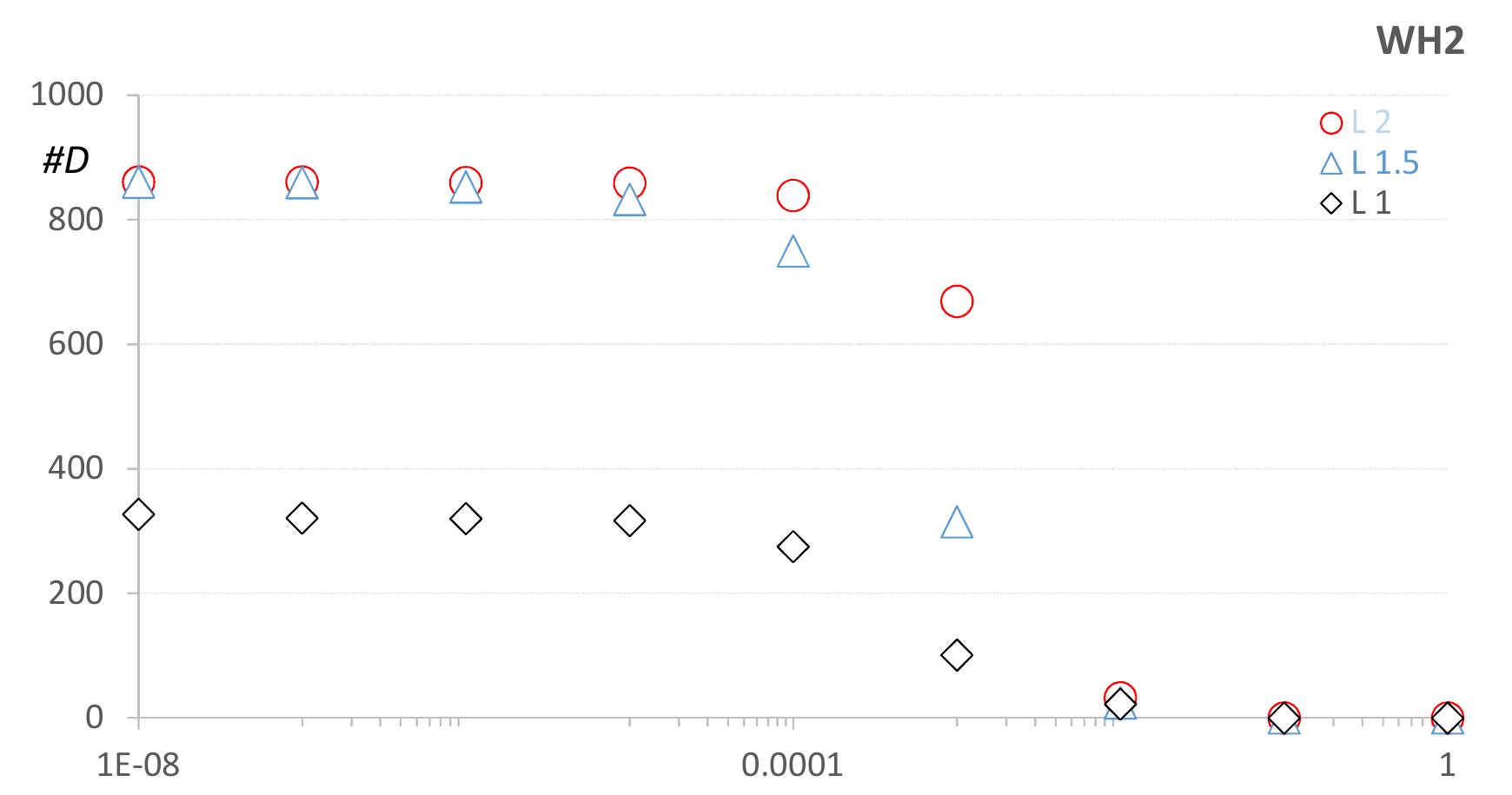}
	\caption{Model sparsities; $D$ = 861, SNR = 40, $N$ = 1000.}
\fi
    \end{subfigure}
\caption{WH2 system. Comparison of model errors and their sparsities.}
\label{Fig_WH2_errors}
\end{figure}

The performance of the proposed algorithms was finally compared to the
algorithms based on Volterra series and unconstrained least squares
approach. The results are presented in Fig. \ref{Fig_Lq_LS}. One can observe
that the models obtained by the proposed constrained algorithms (especially
for $q=1$) have smaller error even when the number of measurements is
significantly smaller ($N=500$ vs. $N=5000$) that in the unconstrained case.
Comparing our results to those in the literature, we would like to point out
that:

\begin{enumerate}
\item In the \cite{WachSli:2015}, where the $l_{1}$ algorithm was originally
proposed, the error of the model was compared to the best possible model
(that is, to the best approximation of the system by the Volterra series).
The difference was apparently small; see Fig. 2 in \cite{WachSli:2015}. We
would also like to point out that the experiments made in \cite{WesKea:2003} \&\ 
\cite{KekGia:2011} were based on systems with a relatively short memory so that
the resulting models with $L=11$ and $P=3$ were of size $D=364$. In our
benchmark experiment, in turn, the model with $D=2441$\ (\emph{i.e.} of
order of magnitude larger)\ was used.

\item In the \cite{Wach:2016}, where the $l_{1}$ algorithm was tested
against the infinite memory nonlinear system, the comparison, similar to the
added to the corrected version of the manuscript, was presented. The
modeled system was again much simpler than the Wiener-Hammerstein benchmark
examined in our manuscript and the resulting Volterra models were
subsequently significantly smaller in size. Nevertheless, the advantage of
the $l_{1}$ algorithm over the $LS$ ones is also noticeable there; see Fig.
3 in \cite{Wach:2016}. 
\begin{figure}[tbp]
\centering
\ifpicture
	\includegraphics[width=0.5\textwidth]{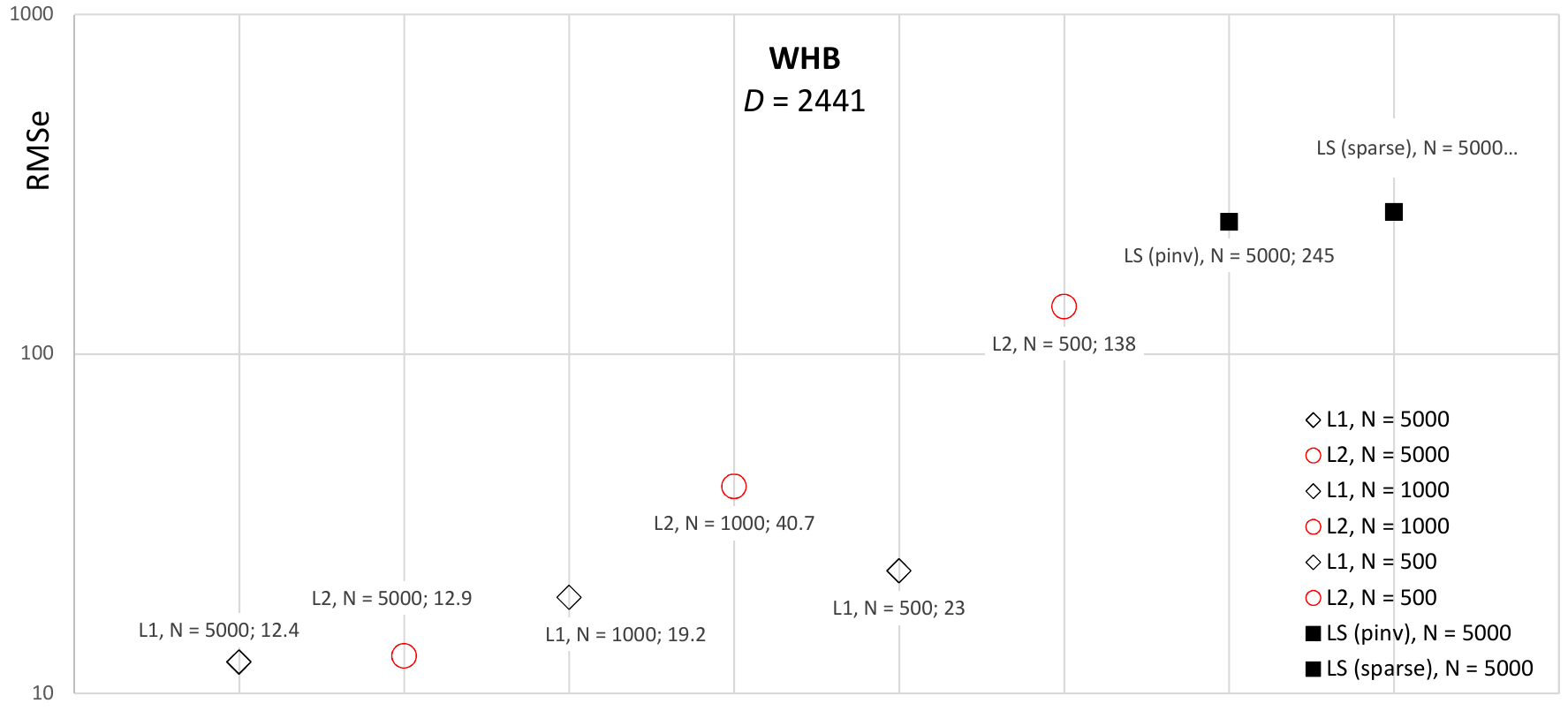} 
\fi
\caption{Performance of the proposed constrained algorithms versus the
unconstrained ones}
\label{Fig_Lq_LS}
\end{figure}
\end{enumerate}

\section{Conclusions}

The problem of effective modeling of nonlinear dynamic systems using short
data records remains both interesting from the formal viewpoint and
important in practice. In the paper we examined the modeling algorithm based
on Volterra series and convex optimization, and concluded that:

\begin{itemize}
\item The error between the resulting empirical models and their best
possible theoretical counterparts is no larger than $\mathcal{O}(\sqrt{\ln D}%
/\sqrt{N})$ for any finite $D>2$. This makes the algorithm robust against
overparametrization and superior to the unconstrained least squares
algorithms, especially for $D\approx \sqrt{N}$; and, \emph{a fortiori}, for $%
D\geq N$.

\item The benchmark experiments confirmed such robustness for $q=1$ in
particular and, to less extent, also for $q>1$. For $q=1$, the algorithm
prefers sparse (\textquotedblleft non-smooth\textquotedblright ) solutions
at the expense of its system nature modeling capability. On the other hand,
the models obtained for $q=1.5,2$, while worse in terms of model error, are
more accurate in reproducing the nature of the system.
\end{itemize}

The most apparent and challenging issue of the proposed approach is the
large size of the Volterra models. This is however the consequence of a poor
prior knowledge about the system structure and the variety of possible
nonlinear structures and characteristics. It should be noted that if more
prior knowledge is available then the more effective and parsimonious
algorithms can be applied, however, such algorithms are specifically
designed for particular system structures (be it Hammerstein, or Wiener, or
Wiener-Hammerstein systems) and fail if the actual system is different than
assumed; see \emph{e.g.} \cite{Will:2012,WachMzyk:2016}.

\appendix

\section*{Proof of Theorem \protect\ref{Th-1}}

The following derivation is based on the proof in \cite{JudNem:2000} and
extends the one in \cite{WachSli:2015} for the case when $q>1$ by employing
the norm inequality in (\ref{norm_ineq}). We use the matrix/vector notation%
\begin{equation*}
\mathbf{A}=E\left\{ \mathbf{s}_{i}\mathbf{s}_{i}^{T}\right\} ,\mathbf{b}%
=2E\left\{ \mathbf{s}_{i}y_{i}\right\} ,\text{ and }c=E\left\{
y_{i}^{2}\right\} ,
\end{equation*}%
where the vectors $\mathbf{s}_{i}$, $i=\tau +1,\ldots ,N,$ are defined as in
(\ref{sh}), so that $Q\left( \mathbf{\theta }\right) $ in (\ref{Q_opt})\ can
be rewritten into a form 
\begin{equation*}
Q\left( \mathbf{\theta }\right) =\mathbf{\theta }^{T}\mathbf{A\theta }-%
\mathbf{\theta }^{T}\mathbf{b}+c=Q_{0}\left( \mathbf{\theta }\right) +c.
\end{equation*}%
where $Q_{0}\left( \mathbf{\theta }\right) =\mathbf{\theta }^{T}\mathbf{%
A\theta }-\mathbf{\theta }^{T}\mathbf{b}$. Analogously, we define 
\begin{equation*}
\mathbf{\hat{A}}=\frac{1}{N-\tau }\sum_{i=\tau +1}^{N}\mathbf{s}_{i}\mathbf{s%
}_{i}^{T}.
\end{equation*}%
Note that $\mathbf{\hat{A}}$, being an empirical counterpart of $\mathbf{A}$%
, is also a symmetric $D\times D$ matrix with the components%
\begin{equation*}
\mathbf{\hat{A}}_{lk}=\frac{1}{N-\tau }\sum_{i=\tau +1}^{N}m_{l}\left( 
\mathbf{u}_{i}\right) m_{k}\left( \mathbf{u}_{i}\right) ,
\end{equation*}%
for $k,l=1,\ldots ,D$. Defining similarly the other two vectors%
\begin{equation*}
\mathbf{\hat{b}}=\frac{2}{N-\tau }\sum_{i=\tau +1}^{N}\mathbf{s}_{i}y_{i}~%
\text{and }\hat{c}=\frac{1}{N-\tau }\sum_{i=\tau +1}^{N}y_{i}^{2}.
\end{equation*}%
we get the following counterpart of (\ref{Q_rand}) 
\begin{equation*}
\hat{Q}\left( \mathbf{\theta }\right) =\mathbf{\theta }^{T}\mathbf{\hat{A}%
\theta }-\mathbf{\theta }^{T}\mathbf{\hat{b}}+\hat{c}=\hat{Q}_{0}\left\{ 
\mathbf{\theta }\right\} +\hat{c},
\end{equation*}%
where 
\begin{equation*}
\hat{Q}_{0}\left( \mathbf{\theta }\right) =\mathbf{\theta }^{T}\mathbf{\hat{A%
}\theta }-\mathbf{\theta }^{T}\mathbf{\hat{b}.}
\end{equation*}

To show the bound for $E\{Q(\mathbf{\hat{\theta}})\}-Q\left( \mathbf{\theta }%
^{\ast }\right) $, we will use the inequality 
\begin{eqnarray*}
Q(\mathbf{\hat{\theta})}-Q\left( \mathbf{\theta }^{\ast }\right) \hspace*{%
-0.05in}\hspace*{-0.05in} &=&\hspace*{-0.05in}\hspace*{-0.05in}Q_{0}(\mathbf{%
\hat{\theta})}-Q_{0}\left( \mathbf{\theta }^{\ast }\right) \\
\hspace*{-0.05in}\hspace*{-0.05in} &=&\hspace*{-0.05in}\hspace*{-0.05in}%
\left[ Q_{0}(\mathbf{\hat{\theta})}-\hat{Q}_{0}\left( \mathbf{\hat{\theta}}%
\right) \right] \hspace*{-0.05in}+\hspace*{-0.05in}\left[ \hat{Q}_{0}\left( 
\mathbf{\hat{\theta}}\right) -Q_{0}\left( \mathbf{\theta }^{\ast }\right) %
\right] \\
\hspace*{-0.05in}\hspace*{-0.05in} &\leq &\hspace*{-0.05in}\hspace*{-0.05in}%
2\sup_{\left\Vert \mathbf{\theta }\right\Vert _{q}\leq D^{1/q-1}}\left\vert 
\hat{Q}_{0}\left( \mathbf{\theta }\right) -Q_{0}\left( \mathbf{\theta }%
\right) \right\vert ,
\end{eqnarray*}%
which allows us to get rid from the analysis the empirical parameter vector $%
\mathbf{\hat{\theta}}$, as it now remains on the left-hand side only. The
term in the right-hand side of the above inequality can further be
decomposed as 
\begin{eqnarray*}
\hat{Q}_{0}\left( \mathbf{\theta }\right) -Q_{0}\left( \mathbf{\theta }%
\right) &=&\mathbf{\theta }^{T}\mathbf{\hat{A}\theta }-\mathbf{\theta }^{T}%
\mathbf{\hat{b}}-\left[ \mathbf{\theta }^{T}\mathbf{A\theta }-\mathbf{\theta 
}^{T}\mathbf{b}\right] \\
&=&\mathbf{\theta }^{T}\left( \mathbf{\hat{A}}-\mathbf{A}\right) \mathbf{%
\theta }-\mathbf{\theta }^{T}\left( \mathbf{\hat{b}}-\mathbf{b}\right) .
\end{eqnarray*}

Taking now into account that the constraint $\left\Vert \mathbf{\theta }%
\right\Vert _{1}\leq 1$ holds by virtue of the assumption in (\ref{argmin})
and because of the norm inequality in (\ref{norm_ineq}), and by applying
both triangle and H\"{o}lder inequalities to $\left\vert \mathbf{\theta }%
^{T}\left( \mathbf{\hat{A}}-\mathbf{A}\right) \mathbf{\theta }\right\vert \ $%
and to $\left\vert \mathbf{\theta }^{T}\left( \mathbf{\hat{b}}-\mathbf{b}%
\right) \right\vert $, we get that%
\begin{equation*}
\left\vert \hat{Q}_{0}\left( \mathbf{\theta }\right) -Q_{0}\left( \mathbf{%
\theta }\right) \right\vert \leq 2\left\Vert \mathbf{\gamma }\right\Vert
_{\infty },
\end{equation*}%
where 
\begin{equation*}
\left\Vert \mathbf{\gamma }\right\Vert _{\infty }=\max_{i}\left\vert \mathbf{%
\gamma }_{i}\right\vert ,\mathbf{\gamma }=[a_{11},\dots ,a_{lk},b_{1},\dots
,b_{D}]^{T},
\end{equation*}%
is an auxiliary vector composed of the unique components $\left\{
a_{lk}\right\} ,$ $1\leq l\leq k\leq D,$ of the matrix $\mathbf{\hat{A}}-%
\mathbf{A}$ and of components $\left\{ b_{l}\right\} $ of the vector $%
\mathbf{\hat{b}}-\mathbf{b}$, respectively. Observing now that\textbf{\ }$%
\mathbf{\gamma }$ can be further rewritten as%
\begin{equation*}
\mathbf{\gamma =}\frac{1}{N-\tau }\sum_{i=\tau +1}^{N}\mathbf{\eta }_{i},
\end{equation*}%
where 
\begin{equation*}
\mathbf{\eta }_{i}\hspace*{-0.02in}=\hspace*{-0.02in}[p_{11}\hspace*{-0.02in}%
\left( \mathbf{u}_{i}\right) \text{,}\dots \text{,}p_{lk}\hspace*{-0.02in}%
\left( \mathbf{u}_{i}\right) \text{,}\dots \text{,}p_{DD}\hspace*{-0.02in}%
\left( \mathbf{u}_{i}\right) \text{,}q_{1}\hspace*{-0.02in}\left( \mathbf{u}%
_{i}\right) \text{,}\dots \text{,}q_{D}\hspace*{-0.02in}\left( \mathbf{u}%
_{i}\right) ]^{T}
\end{equation*}%
for $1\leq l\leq k\leq D$, with%
\begin{equation*}
\begin{array}{l}
p_{lk}\left( \mathbf{u}_{i}\right) =m_{l}\left( \mathbf{u}_{i}\right)
m_{k}\left( \mathbf{u}_{i}\right) -E\left\{ m_{l}\left( \mathbf{u}%
_{0}\right) m_{k}\left( \mathbf{u}_{0}\right) \right\} , \\ 
q_{l}\left( \mathbf{u}_{i}\right) =2m_{l}\left( \mathbf{u}_{i}\right) \left[
m\left( \mathbf{u}_{i}\right) +e_{i}\right] -2E\left\{ m_{l}\left( \mathbf{u}%
_{0}\right) m\left( \mathbf{u}_{i}\right) \right\} ,%
\end{array}%
\end{equation*}%
we get the bound%
\begin{eqnarray*}
E\{Q(\mathbf{\hat{\theta})\}}-Q\left( \mathbf{\theta }^{\mathbf{\ast }%
}\right) &\leq &2E\left\{ \sup_{\left\Vert \mathbf{\theta }\right\Vert
_{1}\leq 1}\left\vert \hat{Q}_{0}\left( \mathbf{\theta }\right) -Q_{0}\left( 
\mathbf{\theta }\right) \right\vert \right\} \\
&\leq &\frac{4}{N-\tau }E\left\Vert \sum_{i=\tau +1}^{N}\mathbf{\eta }%
_{i}\right\Vert _{\infty }.
\end{eqnarray*}

Recall now that the map $m\left( \cdot \right) $ and the dictionary elements
are bounded by $M=\max \left\{ M_{m},M_{d}\right\} $ (\emph{cf.} Assumption 
\textbf{A3}), and hence%
\begin{equation}
\left\vert p_{lk}\right\vert \left. \leq \right. 2M^{2}\text{ and }%
\left\vert q_{l}\right\vert \leq 2M\left\vert e_{i}\right\vert +4M^{2},
\label{M}
\end{equation}%
for all $1\leq l\leq k\leq D$ and $i=\tau +1,\ldots ,N$. Following from this
point the proof as in \cite{WachSli:2015}, we will eventually get the bound%
\begin{equation*}
E\{Q(\mathbf{\hat{\theta}})\}-Q\left( \mathbf{\theta }^{\ast }\right) \leq C%
\frac{\sqrt{N}}{N-\tau }\sqrt{\left( \tau +1\right) \ln D},
\end{equation*}%
with the constant $C=32\sqrt{e}\left( M\sigma +M^{2}\right) $, as in (\ref%
{CV_psi_II}).

\begin{remark}
\label{MR_proof}To take into account the multiplication factor $R$,
introduced in Algorithm \ref{Algorithm_tuning}, and to get the constant $C$
as in (\ref{CV_psi_III}), we only need to replace the bounds in (\ref{M}) by
the following ones%
\begin{equation}
\left\vert p_{lk}\right\vert \left. \leq \right. 2\left( MR\right) ^{2}\text{
and }\left\vert q_{l}\right\vert \leq 2MR\left\vert e_{i}\right\vert
+4(MR)^{2},  \label{MR}
\end{equation}%
for all $1\leq l\leq k\leq D$ and $i=\tau +1,\ldots ,N$.
\end{remark}

\section*{Acknowledgment}

This work was supported in part by the Fund for Scientific Research
(FWO-Vlaanderen), by the Flemish Government (Methusalem), by the Belgian
Government through the Inter university Poles of Attraction (IAP VII)
Program, and by the ERC Advanced Grant SNL-SID, under contract 320378,
and by the Wroc\l aw University of Science and Technology Grants 0401/0217/16 and S50198.

All Authors would also like to thank Prof. Johan Schoukens for his inspiring
suggestions and helpful comments and to Reviewers for their insightful remarks.


\end{document}